\documentclass{llncs}
\usepackage{amssymb,amsmath}
\usepackage{mathrsfs}  
\usepackage{graphicx}
\usepackage{dcolumn}
\usepackage{bm}
\usepackage{color}
\usepackage[all]{xy}
\usepackage{ulem}
\begin{document}
%
%
\newcommand{\beq}{\begin{equation}}
\newcommand{\beqa}{\begin{eqnarray}}
\newcommand{\eeq}{\end{equation}}
\newcommand{\eeqa}{\end{eqnarray}}
\newcommand{\non}{\nonumber}
\newcommand{\fr}[1]{(\ref{#1})}
\newcommand{\cc}{\mbox{c.c.}}
\newcommand{\nr}{\mbox{n.r.}}
\newcommand{\eq}{\mathrm{eq}}
\newcommand{\B}{\mathrm{B}}
\newcommand{\M}{\mathrm{M}}
\newcommand{\Ising}{\mathrm{Ising}}
\newcommand{\heatbath}{\mathrm{heat\, bath}}
\newcommand{\atan}{\mathrm{Tan}{}^{-1}}
\newcommand{\atanh}{\mathrm{Tanh}{}^{-1}}
\newcommand{\acosh}{\mathrm{Cosh}{}^{-1}}
\newcommand{\bb}{\mbox{\boldmath {$b$}}}
\newcommand{\bbe}{\mbox{\boldmath {$e$}}}
\newcommand{\bt}{\mbox{\boldmath {$t$}}}
\newcommand{\bn}{\mbox{\boldmath {$n$}}}
\newcommand{\br}{\mbox{\boldmath {$r$}}}
\newcommand{\bC}{\mbox{\boldmath {$C$}}}
\newcommand{\bH}{\mbox{\boldmath {$H$}}}
\newcommand{\bp}{\mbox{\boldmath {$p$}}}
\newcommand{\bx}{\mbox{\boldmath {$x$}}}
\newcommand{\bF}{\mbox{\boldmath {$F$}}}
\newcommand{\bT}{\mbox{\boldmath {$T$}}}
\newcommand{\bomega}{\mbox{\boldmath {$\omega$}}}
\newcommand{\ve}{{\varepsilon}}
\newcommand{\e}{\mathrm{e}}
\newcommand{\F}{\mathrm{F}}
\newcommand{\Loc}{\mathrm{Loc}}
\newcommand{\Ree}{\mathrm{Re}}
\newcommand{\Imm}{\mathrm{Im}}
\newcommand{\hF}{\widehat F}
\newcommand{\hL}{\widehat L}
\newcommand{\tA}{\widetilde A}
\newcommand{\tB}{\widetilde B}
\newcommand{\tC}{\widetilde C}
\newcommand{\tL}{\widetilde L}
\newcommand{\tK}{\widetilde K}
\newcommand{\tX}{\widetilde X}
\newcommand{\tY}{\widetilde Y}
\newcommand{\tU}{\widetilde U}
\newcommand{\tZ}{\widetilde Z}
\newcommand{\talpha}{\widetilde \alpha}
\newcommand{\te}{\widetilde e}
\newcommand{\tv}{\widetilde v}
\newcommand{\tx}{\widetilde x}
\newcommand{\ty}{\widetilde y}
\newcommand{\ud}{\underline{\delta}}
\newcommand{\uD}{\underline{\Delta}}
\newcommand{\chN}{\check{N}}
\newcommand{\cA}{{\cal A}}
\newcommand{\cB}{{\cal B}}
\newcommand{\cC}{{\cal C}}
\newcommand{\cD}{{\cal D}}
\newcommand{\cF}{{\cal F}}
\newcommand{\cI}{{\cal I}}
\newcommand{\cL}{{\cal L}}
\newcommand{\cM}{{\cal M}}
\newcommand{\cN}{{\cal N}}
\newcommand{\cO}{{\cal O}}
\newcommand{\cP}{{\cal P}}
\newcommand{\cR}{{\cal R}}
\newcommand{\cS}{{\cal S}}
\newcommand{\cY}{{\cal Y}}
\newcommand{\cZ}{{\cal Z}}
\newcommand{\cU}{{\cal U}}
\newcommand{\cV}{{\cal V}}
\newcommand{\dr}{\mathrm{d}}
\newcommand{\sech}{\mathrm{sech}}
\newcommand{\Exp}{\mathrm{Exp}}
 \newcommand{\GamLamM}[1]{{{\cal S}\,\Lambda^{{#1}}\,\cal{M}}}
 \newcommand{\GTM}{{{\cal S}\, T\cal{M}}}
 \newcommand{\GTC}{{{\cal S}\, T\cal{C}}}
\newcommand{\inp}[2]{\left\langle\,  #1\, , \, #2\, \right\rangle}
\newcommand{\equp}[1]{\overset{\mathrm{#1}}{=}}
\newcommand{\wt}[1]{\widetilde{#1}}
\newcommand{\wh}[1]{\widehat{#1}}
\newcommand{\ch}[1]{\check{#1}}
\newcommand{\ii}{\imath}
\newcommand{\ic}{\iota}
\newcommand{\scrH}{\mathscr{H}}
\newcommand{\mi}{\,\mathrm{i}\,}
\newcommand{\mr}{\,\mathrm{r}\,}
\newcommand{\mbbC}{\mathbb{C}}
\newcommand{\mbbE}{\mathbb{E}}
\newcommand{\mbbR}{\mathbb{R}}
\newcommand{\mbbZ}{\mathbb{Z}}
\newcommand{\ol}[1]{\overline{#1}}
\newcommand{\rmC}{\mathrm{C}}
\newcommand{\rmH}{\mathrm{H}}
\newcommand{\Id}{\mathrm{Id}} 
\newcommand{\avg}[1]{\left\langle\,{#1}\, \right\rangle}
\newcommand{\Leg}{\mathbb{L}}
\newcommand{\avgg}[1]{\left\langle\langle\,{#1}\, \rangle\right\rangle}
\title{
Expectation variables
on a para-contact metric manifold  
 exactly derived from master equations
}
\author{  Shin-itiro Goto\orcidID{0000-0002-5249-1054} \and  
Hideitsu Hino\orcidID{0000-0002-6405-4361}} 
\authorrunning{S. Goto and H. Hino}
\institute{
The Institute of Statistical Mathematics,  \\
10-3 Midori-cho, Tachikawa, Tokyo 190-8562, Japan
\qquad 
}
\maketitle
\begin{abstract}%
Based on information and para-contact metric geometries, 
in this paper
a class of dynamical systems  is formulated for 
describing time-development of expectation variables. 
Here  such systems for expectation variables are 
exactly derived from 
continuous-time master equations describing nonequilibrium processes. 
\end{abstract}%
\section{Introduction}
Information geometry is a geometrization of mathematical 
statistics \cite{AN,Ay2017}, 
and its differential geometric aspects and 
applications in statistics have been investigated. 
Examples of applications of information geometry include 
thermodynamics,
and some links between equilibrium thermodynamics and 
 information geometry have been clarified.  
In addition, links  between information geometry and contact geometry 
have been argued \cite{Goto2015,Goto2016}. 
In this context, it was found that para-Sasakian 
geometry is suitable for describing thermodynamics 
\cite{Bravetti2015JPhysA,Bravetti2015}, 
where para-Sasakian manifolds are  para-contact metric 
manifolds satisfying some additional condition.  
We then ask how para-contact metric manifolds describe thermodynamics.

In this paper a class of nonequilibrium thermodynamic processes are 
formulated on a para-contact metric manifold. 
Most of discussions in this paper have been in \cite{GH2018}, 
and 
those involving an almost para-contact structure 
are given in this contribution.  

\section{Preliminaries}
\label{section-preliminaries}
In this paper manifolds are assumed smooth and connected. In addition
tensor fields are assumed smooth 
and real. The set of vector fields on a manifold $\cM$ is 
denoted by $\GTM$, and the Lie derivative along $X\in\GTM$ 
by $\cL_{X}$.

In this section definitions and some existing 
statements are summarized. 
(see \cite{Zamkovoy2009,Bravetti2015JPhysA}).

Let $\cM$ be a $(2n+1)$-dimensional manifold ($n\geq 1$).
An {\it almost para-contact structure} on $\cM$ is a triplet 
$(\phi,\xi,\lambda)$, where $\xi$  
is a vector field, 
$\lambda$  
a one-form, $\phi:\GTM\to\GTM$ a $(1,1)$-tensor field 
such that 
$$
\mbox{(i)}:\ 
\phi^{\,2}=\Id-\lambda\otimes\xi,\ 
\mbox{(ii)}:\ 
\lambda(\xi)=1,
\ \mbox{and}\ 
\mbox{(iii)}:\ 
\ker(\lambda)
=\Imm(\phi)=\cD^{\,+}+\cD^{\,-},
$$
where $\ker(\lambda):=\{X\in \GTM\,|\,\lambda(X)=0\,\}$, 
$\Imm(\phi):=\{\phi(X,-)\in\GTM|\,X\in\GTM\}$, 
$\cD^{\,\pm}$ are eigen-spaces whose eigenvalues are $\pm1$, 
and 
$\Id$ is an identity operator. 
A pseudo Riemannian metric tensor field $g$ satisfying 
$$
g(\phi X,\phi Y)
=-\,g(X,Y)+\lambda(X)\,\lambda(Y),\qquad \forall X,Y\in\GTM 
$$
is referred to as a {\it metric tensor 
compatible with an almost para-contact structure}.
It is verified for non-compact manifolds  
that any almost para-contact structure admits a metric tensor field 
compatible with an almost para-contact structure. 
Then $(\cM,\phi,\xi,\lambda,g)$ is referred to as 
an {\it almost para-contact metric manifold}.

On almost para-contact metric  manifolds, one can show 
that 
$$
\lambda\,\phi
=0,\ 
\phi\,\xi=0,\ 
\lambda(X)
=g(X,\xi),\ 
g(\xi,\xi)=1,\ \mbox{and}\ 
g(\phi X,Y)+g(X,\phi Y)
=0, 
$$
for $\forall X,Y\in\GTM$. 
If $g$ of an almost para-contact metric manifold satisfies 
\beq
g(X,\phi Y)
=\frac{1}{2}\dr\lambda(X,Y),\qquad \forall X,Y\in\GTM
\label{para-Sasakian-condition-1}
\eeq
then, $(\cM,\phi,\xi,\lambda,g)$ is referred to as a 
{\it para-contact metric manifold}, 
where the convention  of the numerical factor  
$\dr\lambda(X,Y)=X\lambda(Y)-Y\lambda(X)-\lambda([X,Y])$
, ($[X,Y]:=XY-YX$) has been adopted.
Para-Sasakian manifolds are para-contact metric manifolds 
satisfying the so-called normality condition.

Coordinate expressions were 
given for a  para-contact metric manifold  
( and a para-Sasakian manifold ) $(\cM,\phi,\xi,\lambda,g)$ 
in \cite{Bravetti2015JPhysA}. 
They are summarized here.
Let $(x,y,z)$ be coordinates for $\cM$ with 
$x=\{x^{\,1},\ldots,x^{\,n}\}$ and $y=\{y_{\,1},\ldots,y_{\,n}\}$ such that 
$\lambda=\dr z-y_{\,a}\dr x^{\,a}$ where the Einstein convention has been used.
Introduce the pseudo-Riemannian metric tensor field, referred to as the 
{\it Mruagala metric tensor field} \cite{Mrugala1990},  
\beq
g^{\,\M}=
\frac{1}{2}\dr x^{\,a}\otimes \dr y_{\,a}
+\frac{1}{2}\dr y_{\,a}\otimes \dr x^{\,a}
+\lambda\otimes\lambda,
\label{metric-G-Bravetti}
\eeq
which is shown to induce a para-contact metric manifold.
In what follows we consider the case where $y_{\,a}>0$ for all 
$a\in\{1,\ldots,n\}$.
Introduce the co-frame 
$\{\wh{\theta}^{\,0},\wh{\theta}_{\,-}^{\,1},\wh{\theta}_{\,+}^{\,1},\ldots,
\wh{\theta}_{\,-}^{\,n},\wh{\theta}_{\,+}^{\,n}\}$ 
and frame $\{e_{\,0},e_{\,1}^{\,-},e_{\,1}^{\,+},\ldots,e_{\,n}^{\,-},e_{\,n}^{\,+}\}$ 
with  
\beqa
\wh{\theta}^{\,0}
&:=&\lambda,\qquad 
\wh{\theta}_{\,\pm}^{\,a}
:=\frac{1}{2\sqrt{y_{\,a}}}\left[\,y_{\,a}\dr x^{\,a}\pm \dr y_{\,a}\,\right],
\qquad
\mbox{(no sum over $a$)},
\non\\
e_{\,0}
&:=&\xi,\qquad 
e_{\,a}^{\,\pm}
:=\sqrt{y_{\,a}}\left[\,
\frac{1}{y_{\,a}}\left(\frac{\partial}{\partial x^{\,a}}+y_{\,a}\frac{\partial}{\partial z}\right)\pm \frac{\partial}{\partial y_{\,a}}
\,\right],\qquad\mbox{(no sum over $a$)},
\non
\eeqa
so that 
$$
\wh{\theta}^{\,0}(e_{\,0})
=1,\qquad 
\wh{\theta}_{\,+}^{\,a}(e_{\,b}^{+})
=\wh{\theta}_{\,-}^{\,a}(e_{\,b}^{-})
=\delta_{\,b}^{\,a},\qquad\mbox{others vanish,}
$$ 
where $\delta_{\,b}^{\,a}$ is the Kronecker delta, 
giving unity for $a=b$ and zero otherwise. 
One can then show that 
$$
g^{\,\M}=\wh{\theta}^{\,0}\otimes\,\wh{\theta}^{\,0}
+\sum_{a=1}^{n}\wh{\theta}_{\,+}^{\,a}\otimes\,\wh{\theta}_{\,+}^{\,a}
-\sum_{a=1}^{n}\wh{\theta}_{\,-}^{\,a}\otimes\,\wh{\theta}_{\,-}^{\,a},
\quad
\xi=\frac{\partial}{\partial z},\quad 
\phi=-\,\wh{\theta}_{\,-}^{\,a}\otimes e_{\,a}^{\,+}
-\,\wh{\theta}_{\,+}^{\,a}\otimes e_{\,a}^{\,-}.
$$
Then, introducing the abbreviation $\phi(X):=\phi(X,-)\in\GTM$ 
for $X\in\GTM$, one has
$\phi(\, e_{\,a}^{\,+}\,)=-\,e_{\,a}^{\,-}$,
$\phi(\, e_{\,a}^{\,-}\,)=-\,e_{\,a}^{\,+}$, and 
$\phi(\, e_{\,0}\,)=0$.

In the context of geometry of thermodynamics, 
contact manifold 
is identified with the so-called 
thermodynamic phase space \cite{Mrugala2000}. 
This manifold 
 is defined as follows (see \cite{Silva2008} for details).
Let $\cC$ be a $(2n+1)$-dimensional manifold ($n=1,2,\ldots$), and $\lambda$ a one-form.  
If $\lambda$
satisfies 
$$
\lambda\wedge\underbrace{\dr\lambda\wedge\cdots\wedge\dr\lambda}_{n}\neq 0,
$$
then the pair $(\cC,\lambda)$ is referred to as a {\it contact manifold}, 
and $\lambda$ a {\it contact one-form}. 
It has been known as the 
Darboux theorem  that there exists a special set of coordinates 
$(x,y,z)$ with $x=\{x^{\,1},\ldots,x^{\,n}\}$ and $y=\{y_{\,1},\ldots,y_{\,n}\}$  
such that $\lambda=\dr z-y_{\,a}\dr x^{\,a}$. 
It follows from \fr{para-Sasakian-condition-1} that para-contact metric 
manifolds are contact manifolds.

The {\it Legendre submanifold} $\cA\subset\cC$ 
is an $n$-dimensional submanifold where $\lambda|_{\,\cA}=0$ holds.
One can verify that 
\beq
\cA_{\,\varpi}
=\left\{\,(x,y,z)\ \bigg|\ 
y_{\,a}
=\frac{\partial\varpi}{\partial x^{\,a}},\quad
\mbox{and}\quad 
z=\varpi(x)\,\right\}, 
\label{Legendre-submanifold-psi}
\eeq
is a Legendre submanifold, where $\varpi:\cC\to\mbbR$  is a function of $x$ 
on $\cC$. The submanifold $\cA_{\,\varpi}$ is referred to as the 
{\it Legendre submanifold generated by $\varpi$}, and is used for describing 
equilibrium thermodynamic systems \cite{Mrugala2000}. 

As shown in \cite{Goto2015} and \cite{Bravetti2015},  
a class of relaxation processes, 
initial states approach to the equilibrium state as time develops, 
can be formulated as contact Hamiltonian vector fields on contact manifolds.  
This statement on a class of contact Hamiltonian vector fields 
can be summarized as follows.
\begin{proposition}
\label{proposition-goto-2015}
(Legendre submanifold as an attractor, \cite{Goto2015}). 
Let $(\cC,\lambda)$ be a $(2n+1)$-dimensional contact manifold with $\lambda$ 
being a contact form, $(x,y,z)$ its coordinates so that 
$\lambda=\dr z-y_{\,a}\dr x^{\,a}$,  
and 
$\varpi$ a function
depending only on $x$. Then, one has 
\begin{enumerate}
\item
The contact Hamiltonian vector field associated with 
the contact Hamiltonian $h:\cC\to\mbbR$ such that
$h(x,y,z)=\varpi(x)-z,$
gives 
\beq
\frac{\dr}{\dr t}x^{\,a}
=0,\qquad 
\frac{\dr}{\dr t}y_{\,a}
=\frac{\partial\,\varpi}{\partial x^{\,a}}-y_{\,a},\qquad 
\frac{\dr}{\dr t}z
=\varpi(x)-z.
\label{contact-hamilton-flow-pi}
\eeq
\item
The Legendre submanifold generated by $\varpi$, given by 
\fr{Legendre-submanifold-psi}, is an invariant manifold for 
the contact Hamiltonian vector field.
\item
Every point on $\cC\setminus\cA_{\,\varpi}$ approaches to $\cA_{\,\varpi}$ along 
an integral curve as time develops. 
Equivalently $\cA_{\,\varpi}$ is an attractor in $\cC$.  
\item 
Let $\{x(0),y(0),z(0)\}$ be a point on $\cC\setminus\cA_{\,\varpi}$.
Then for any $t\in\mbbR$,  
$$
h(x(t),y(t),z(t))
=\exp(-\,t)\, h(x(0),y(0),z(0)).
$$
\end{enumerate}
\end{proposition}

\section{Solvable master equations}
\label{section-solvable-heat-bath}
In this section 
a set of master equations with particular Markov kernels
is introduced, 
and then its solvability is shown. 

Let $\Gamma$ be a set of finite discrete states,  $t\in\mbbR$ time, and 
$p(j,t)\,\dr t$ a probability that a state $j\in \Gamma$ is found in between 
$t$ and $t+\dr t$. 
The first objective is to realize a given distribution function 
$p_{\,\theta}^{\,\eq}$ 
that can be written as 
$$
p_{\,\theta}^{\,\eq}(j)
=\frac{\pi_{\,\theta}(j)}{Z(\theta)}, \qquad
Z(\theta)
:=\sum_{j\in \Gamma}\pi_{\theta}(j)
$$
where $\theta\in\Theta\subset\mbbR^{\,n}$ is a parameter set with 
$\theta=\{\theta^{\,1},\ldots,\theta^{\,n}\}$, and 
$Z :\Theta\to\mbbR$  the so-called partition function so that 
$p_{\,\theta}^{\,\eq}$ is normalized  :  
$\sum_{j\in \Gamma}p_{\,\theta}^{\,\eq}(j)=1.$

In what follows, attention is focused on 
a class of master equations.
Let $p:\Gamma\times\mbbR\to \mbbR_{\geq 0}$  
be a time-dependent probability function. 
Then, consider the set of master equations 
\beq
\frac{\partial}{\partial t}p(j,t)
=\sum_{j'(\neq j)}\left[\,
w(j|j^{\,\prime})\,p(j',t)-w(j^{\,\prime}|j)\,p(j,t)
\,\right],
\label{master-equation-general}
\eeq
where $w:\Gamma\times\Gamma\to I$, ($I:=[\,0,1\,]\subset\mbbR$)  
is such that    
$w(j|j^{\,\prime})$ denotes a probability that a state jumps from 
$j^{\,\prime}$ to $j$. 
With \fr{master-equation-general} and the assumptions 
$$
w_{\,\theta}(j|j^{\,\prime})
=p_{\,\theta}^{\,\eq}(j),
\qquad\mbox{and}\qquad
p_{\,\theta}^{\,\eq}(j)\neq 0,\quad
\forall j\in\Gamma,
$$
one derives the {\it solvable master equations} :  
\beq
\frac{\partial}{\partial t}p(j,t)
=p_{\,\theta}^{\,\eq}(j)-p(j,t).
\label{master-equation-solvable}
\eeq
An explicit form of $p(j,t)$ is obtained by solving \fr{master-equation-solvable}.
Then the following proposition can easily be shown.
\begin{proposition}
\label{fact-master-equation-solution}
(Solutions of the master equations, \cite{GH2018}). 
The solution of \fr{master-equation-solvable} is 
$$
p(j,t)
=\e^{\,-t}\,p(j,0)+(1-\e^{\,-t})p_{\,\theta}^{\,\eq}(j),\quad
\mbox{from which}\quad 
\lim_{t\to\infty}p(j,t)
=p_{\,\theta}^{\,\eq}(j).
$$
\end{proposition}

With this proposition, one notices 
that every solution $p$ depends on $\theta$. 
Taking into account this, 
$p(j,t)$ is denoted $p(j,t;\theta)$. 
Also notice that the equilibrium state is realized 
with \fr{master-equation-solvable} as the time-asymptotic limit.
\section{Time-development of observables}
\label{section-time-development-observables}
In this section differential equations describing 
time-development of 
observables are derived with the solvable master equations 
 under some assumptions. 
 Then, the time-asymptotic limit of such observables is stated.    
Here {\it observable} in this paper 
is defined as a function that does not depend on 
a random variable or a state.
Thus expectation values with respect to a probability distribution function  
are observables. 

Let $\cO_{\,a}:\Gamma\to \mbbR$ be a function 
with $a\in\{1,\ldots,n\}$, and 
$p:\Gamma\times\mbbR\to \mbbR_{\geq 0}$  
a distribution function that 
follows \fr{master-equation-solvable}. Then 
$$
\avg{\cO_{\,a}}_{\,\theta}(t)
:=\sum_{j\in \Gamma}\cO_{\,a}(j)\,p(j,t;\theta),\qquad \mbox{and}\qquad
\avg{\cO_{\,a}}_{\,\theta}^{\,\eq}
:=\sum_{j\in \Gamma}\cO_{\,a}(j)\,p_{\,\theta}^{\,\eq}(j),
$$
are referred to as the {\it expectation variable} of $\cO_{\,a}$ 
with respect to $p$,  
and that with respect to $p_{\,\theta}^{\,\eq}$, respectively.

If 
an equilibrium 
distribution function belongs to the 
exponential family, then the function $\Psi^{\,\eq}:\Theta\to\mbbR$ with
\beq
\Psi^{\,\eq}(\theta)
:=\ln\left(\,\sum_{j\in \Gamma}\e^{\,\theta^{\,b}\cO_{\,b}(j)}
\,\right),
\label{definition-Psi-eq}
\eeq
plays various roles. Here and in what follows, \fr{definition-Psi-eq} 
is assumed to exist.
In the context of information geometry, this function is referred to as 
a {\it $\theta$-potential}. 
Discrete distribution functions 
are considered in this paper 
and it has been known that such distribution functions belong to 
the exponential family, then $\Psi^{\,\eq}$ in 
\fr{definition-Psi-eq} also plays a role throughout this paper.   
The value $\Psi^{\,\eq}(\theta)$ can be interpreted as the 
negative dimension-less free-energy. 
It follows from \fr{definition-Psi-eq} that 
$$
\avg{\cO_{\,a}}_{\,\theta}^{\,\eq}
=\frac{\partial\Psi^{\,\eq}}{\partial\theta^{\,a} }.
$$

One then can generalize $\Psi^{\,\eq}$ 
defined at equilibrium state  
to a function defined in  
nonequilibrium states as 
$\Psi:\Theta\times\mbbR\to\mbbR$, 
$$
\Psi(\theta,t)
:=\left(\frac{1}{J^{\,0}}
\sum_{j\in \Gamma}\frac{p(j,t;\theta)}{p_{\theta}^{\,\eq}(j)}\right)
\Psi^{\,\eq}(\theta),\qquad\mbox{where}\qquad
J^{\,0}
:=\sum_{j'\in \Gamma} 1. 
$$
Since $p_{\,\theta}^{\,\eq}(j)\neq 0$ and $\Psi^{\,\eq}(\theta)<\infty$  
by assumptions, the function $\Psi$ exists. 
Generalizing the idea for the equilibrium case, the function $\Psi$ may be 
interpreted as a nonequilibrium negative dimension-less free-energy. 

A set of differential equations for $\{\avg{\cO_{\,a}}_{\,\theta}\}$ and $\Psi$
can be derived as follows.
\begin{proposition}
\label{moment-dynamics}
(Dynamical system obtained from the master equations, \cite{GH2018}).    
Let $\theta$ be a time-independent parameter set characterizing a 
discrete 
distribution function $p_{\,\theta}^{\,\eq}$.
Then $\{\avg{\cO_{\,a}}_{\,\theta}\}$ and $\Psi$ are solutions to 
the differential equations on $\mbbR^{\,2n+1}$
$$
\frac{\dr}{\dr t}\theta^{\,a}
=0,\qquad
\frac{\dr}{\dr t}
\avg{\cO_{\,a}}_{\,\theta}
=-\,\avg{\cO_{\,a}}_{\,\theta}
+\frac{\partial\,\Psi^{\,\eq}}{\partial\theta^{\,a}},\quad\mbox{and}\quad 
\frac{\dr}{\dr t}\Psi
=-\,\Psi+\Psi^{\,\eq}.
$$
\end{proposition}
\begin{remark}
The explicit time-dependence for this system is obtained as 
$\theta^{\,a}(t)=\theta^{\,a}(0)$, and $\Psi(\theta,t)
=\e^{\,-\,t}\left[\,\Psi(0)-\Psi^{\,\eq}(\theta)\,\right]
+\Psi^{\,\eq}(\theta)$, and 
$$
\avg{\cO_{\,a}}_{\,\theta}(t)
=\e^{\,-\,t}\left[\,\avg{\cO_{\,a}}_{\,\theta}(0)
-\frac{\partial\Psi^{\,\eq}}{\partial \theta^{\,a}}\,\right]
+\frac{\partial\Psi^{\,\eq}}{\partial \theta^{\,a}}.
$$
From these,
one can verify that the time-asymptotic limit 
of these variables are those defined at equilibrium. 
In this paper this dynamical system is referred to as the 
{\it moment dynamical system}. 
\end{remark}

\section{Geometric description of dynamical systems }
\label{section-geometry-solvable-extended-heat-bath}
Several geometrization of nonequilibrium states for some 
models and methods  
have been proposed. 
Yet, suffice to say that there
remains no general consensus on how best to extend 
a geometry of 
equilibrium states to a geometry 
of  nonequilibrium states.
In this section, a geometrization of 
nonequilibrium states is proposed for the 
moment dynamical system.

\subsection{Geometry of equilibrium states}
\label{section-geometry-equilibrium-states}

Equilibrium states 
 are identified with the Legendre submanifolds generated by functions    
 in the context of 
geometric thermodynamics \cite{Mrugala1978,Mrugala2000}.  
Besides, in the context of information geometry, equilibrium states are 
identified with dually flat spaces \cite{AN}.  
Combining these identifications, one has the following. 
\begin{proposition}
(A contact manifold and a strictly convex function induce a dually flat 
space, \cite{Goto2015}).  
Let $(\cC,\lambda)$ be a contact manifold, $(x,y,z)$ a set of coordinates such 
that $\lambda=\dr z-y_{\,a}\dr x^{\,a}$ with $x=\{x^{\,1},\ldots,x^{\,n}\}$ and 
$y=\{y_{\,1},\ldots,y_{\,n}\}$, 
and $\varpi$ a strictly convex function depending only on $x$.
Then, 
 $((\cC,\lambda),\varpi)$ induces 
an  
 $n$-dimensional dually flat space
\end{proposition}

To apply the proposition above to physical systems, 
the coordinate sets $x$ and $y$ are chosen such that  $x^{\,a}$ and $y_{\,a}$ 
form a thermodynamic conjugate pair for each $a$. 
Here it is assumed that 
such thermodynamic variables can be defined even for nonequilibrium states, 
and that they are consistent with those variables defined at equilibrium. 
In addition to this, 
the physical dimension of $\varpi$ should be equal to that of 
$y_{\,a}\,\dr x^{\,a}$. Also $\Psi$ and its Legendre transform are 
chosen as $\varpi$.

\subsection{Geometry of nonequilibrium states}
So far geometry of equilibrium states have been discussed. 
One remaining issue is 
how to give the physical meaning of 
the set outside $\cA_{\,\varpi}$, 
$\cC\setminus\cA_{\,\varpi}$. 
A natural interpretation   
of $\cC\setminus\cA_{\,\varpi}$  
would be some set of nonequilibrium states. We make this interpretation 
in this paper.

As shown in proposition\, \ref{fact-master-equation-solution}, 
initial states approach to the equilibrium state as time develops.
This can be reformulated on contact manifolds and para-contact metric 
manifolds.  
In the contact geometric framework of nonequilibrium thermodynamics, 
the equilibrium state is identified with a Legendre submanifold. 
Then, as found 
in \cite{Goto2015} and \cite{Bravetti2015},  
some dynamical systems expressing nonequilibrium process 
can be identified with a class of 
contact Hamiltonian vector fields on a contact manifold. 
The above claim also holds on para-contact metric manifolds.

\subsubsection{Geometry of moment dynamical system}
Proposition\,\ref{moment-dynamics} 
is written in a contact geometric language here.
In what follows phase space is identified with a $(2n+1)$-dimensional 
para-contact metric manifold $(\cC,\phi,\xi,\lambda,g^{\,\M})$.

As shown below, the moment dynamical system is a contact Hamiltonian system.
\begin{proposition}
\label{fact-moment-dynamics-contact-Hamiltonian-system}
(Moment dynamical system as a contact Hamiltonian system, \cite{GH2018}). 
The dynamical system in proposition\,\ref{moment-dynamics}
can be written as a contact Hamiltonian system.
\end{proposition}
One is interested in how a $(1,1)$-tensor field $\phi$ plays a role for 
geometric nonequilibrium thermodynamics. 
To give an answer, one needs the following. 
\begin{lemma}
\label{fact-phi-phi-X}
Let $\{\dot{x}_{\,a}\},\{\dot{y}_{\,a}\},\dot{z}$ 
be some functions, and $X_{\,0}$ the vector field 
$$
X_{\,0}
=\dot{x}^{\,a}\frac{\partial}{\partial x^{\,a}}
+\dot{y}_{\,a}\frac{\partial}{\partial y_{\,a}}
+\dot{z}\frac{\partial}{\partial z}.
$$ 
Then, 
$\phi(X_{\,0})$ and $\phi^{\,2}(X_{\,0})$ are calculated as
$$
\phi^{\,\mu}(X_{\,0})
=
(-1)^{\,\mu}\,\dot{x}^{\,a}\left(\,
\frac{\partial}{\partial x^{\,a}}+y_{\,a}\frac{\partial}{\partial z} 
\,\right)+ 
\dot{y}_{\,a}\frac{\partial}{\partial y_{\,a}},\quad
\mu=1,2.
$$
\end{lemma}
\begin{proof}
Throughout this proof, the Einstein convention is not used.
With the local expressions shown 
 in section\,\ref{section-preliminaries}, one 
has 
$$
\wh{\theta}_{\,\pm}^{\,a}(X_{\,0})
=\frac{\sqrt{y_{\,a}}}{2}\dot{x}^{\,a}\pm \frac{\dot{y_{\,a}}}{2\sqrt{y_{\,a}}},
$$
$$
e_{\,a}^{\,+}+e_{\,a}^{\,-}
=\frac{2}{\sqrt{y_{\,a}}}\left(\frac{\partial}{\partial x^{\,a}}
+y_{\,a}\frac{\partial}{\partial z}\right),\quad\mbox{and}\quad
e_{\,a}^{\,+}-e_{\,a}^{\,-}
=2\,\sqrt{y_{\,a}}\frac{\partial}{\partial y_{\,a}}.
$$
Combining these, one has 
\beqa
\phi(X_{\,0})
&=&\sum_{a}\left[\,-\,\wh{\theta}_{\,-}^{\,a}(X_{\,0})\,e_{\,a}^{\,+}
+\wh{\theta}_{\,+}^{\,a}(X_{\,0})\,e_{\,a}^{\,-}\,\right]
\non\\
&=&
\sum_{a}\left[\,-\,\frac{\sqrt{y_{\,a}}}{2}\dot{x}^{\,a}
(\,e_{\,a}^{\,+}+e_{\,a}^{\,-}\,)
+\frac{\dot{y}_{\,a}}{2\,\sqrt{y_{\,a}}}
(\,e_{\,a}^{\,+}-e_{\,a}^{\,-}\,)\right]
\non\\
&=&\sum_{a}\left[-\,\dot{x}^{\,a}\left(
\frac{\partial}{\partial x^{\,a}}+y_{\,a}\frac{\partial}{\partial z}
\right)+\dot{y}_{\,a}
\frac{\partial}{\partial y_{\,a}}\right].
\non
\eeqa
For $\phi^{\,2}(X_{\,0})$, substituting 
$\lambda(X_{\,0})=\dot{z}-\sum_{a}y_{\,a}\dot{x}^{\,a}$ into  
$\phi^{\,2}(X_{\,0})=X_{\,0}-\lambda(X_{\,0})\xi$, one has the desired expression. 
\qed
\end{proof}
Applying this Lemma, one has the following. 
\begin{theorem}
\label{theorem-phi}
(Roles of $\phi$ of $X_{\,h}$ for the moment dynamical system). 
Let $X_{\,h}$ be the contact Hamiltonian vector field in 
Proposition\,\ref{proposition-goto-2015}. Then
$$
\cL_{\,\phi(X_{\,h})}h
=\cL_{\,\phi^{\,2}(X_{\,h})}h
=0.
$$
\end{theorem}
\begin{proof}
Substituting $\dot{x}^{\,a}=0$  into 
$\phi^{\,\mu}(X_{\,0})$ in Lemma\,\ref{fact-phi-phi-X}, one has 
$$
\phi^{\,\mu}(X_{\,h})
=\dot{y}_{\,a}\frac{\partial}{\partial y_{\,a}},
\qquad \mbox{where}\quad
\dot{y}_{\,a}
=\frac{\partial\varpi}{\partial x^{\,a}}-y_{\,a},
\qquad \mu=1,2.
$$
Then, with $\partial h/\partial y_{\,a}=0$, one has
$$
\cL_{\phi^{\,\mu}(X_{\,h})}h
=\left[\phi^{\,\mu}(X_{\,h})\right]\,h
=0,\qquad \mu=1,2.
$$
\qed
\end{proof}
This states that the $h$ is preserved along $\phi^{\,\mu}(X_{\,h})\in\GTC$, 
which should be compared with the case of $\cL_{\,X_{\,h}}h$ : 
$$
\cL_{\,X_{\,h}}h
=-\,\dot{z}
=\,-(\varpi(x)-z)
=-\,h.
$$ 
\subsubsection{Curve  length from the equilibrium state }
In nonequilibrium statistical physics,  
attention is often concentrated on 
how far a state is close to the equilibrium state.
In general, to define and measure  
such a distance in terms of geometric language,     
length of a curve can be used. 
In Riemannian geometry, 
 length is a measure 
for expressing how far given two points are away, where 
these points are  
connected with an integral curve of a vector field on a 
 manifold. 

The following can easily be proven. 
\begin{lemma}
\label{fact-length-Psi}
 (\cite{GH2018}). The length between a state and the equilibrium state for
the moment dynamical system calculated  with 
\fr{metric-G-Bravetti} is 
\beq
l[X_{\,\Psi}]_{\,\infty}^{\,t}
:=\int_{\infty}^{t}\sqrt{g^{\,\M}(X_{\,\Psi},X_{\,\Psi})}\,\dr t
=|\,h(\theta,\avg{\cO}_{\,\theta},\Psi)\,|,
\label{length-Psi-h}
\eeq
where 
$\avg{\cO}_{\,\theta}=\{\avg{\cO_{\,1}}_{\,\theta},\ldots,\avg{\cO_{\,n}}_{\,\theta}\}$, 
$h$ is such that $h(\theta, \avg{\cO}_{\,\theta},\Psi)=\Psi^{\,\eq}(\theta)-\Psi$ 
(see proposition\,\ref{proposition-goto-2015}), 
and $X_{\Psi}$ its corresponding contact Hamiltonian vector field.
Then the convergence rate for \fr{length-Psi-h} is exponential. 
\end{lemma}

Combining Lemma\,\ref{fact-length-Psi} 
and 
discussions in the previous sections, 
one arrives at 
the main theorem in this paper.
\begin{theorem}
\label{theorem-arxiv-2018}
(Geometric description of the expectation variables and 
its convergence). 
The moment dynamical system derived from solvable master equations are 
described on a para-contact metric manifold, and 
its convergence rate associated with 
the metric tensor field \fr{metric-G-Bravetti} is exponential. 
\end{theorem}
\section{Conclusions}
\label{section-conclusion}
This paper has offered a viewpoint  
that expectation variables of the moment dynamical system derived from 
master equations  
can be described 
on a para-contact metric manifold.   
To give a geometric description of these variables  
a contact Hamiltonian vector field has  
been introduced on a para-contact metric manifold. 
Also, roles of the $(1,1)$-tensor field $\phi$ have been clarified 
in this paper ( Theorem\,\ref{theorem-phi} ). 
Then, with the Mrugala  metric tensor field, the convergence rate 
has been shown to be exponential on this para-contact metric manifold\\
( Theorem\,\ref{theorem-arxiv-2018} ).

\section*{Acknowledgments}
The author S.G. is partially supported by JSPS 
(KAKENHI) grant number 19K03635. 
The other author H.H. is partially supported 
by JSPS (KAKENHI) grant number 17H01793.
In addition, both of the authors are partially supported 
by JST CREST JPMJCR1761. 


\begin{thebibliography}{99}
\bibitem{AN} 
  S. Amari and H. Nagaoka,  
{\it Methods of information geometry}, AMS, Oxford University Press, (2000).

\bibitem{Ay2017}
N. Ay et al, {\it Information geometry}, Springer, (2017).  

\bibitem{Goto2015}
S. Goto, 
{\it Legendre submanifolds in contact manifolds as attractors and geometric nonequilibrium thermodynamics}, 
J. Math. Phys., {\bf 56}, 073301 [ 30 pages ], (2015).

\bibitem{Goto2016}
S. Goto, 
{\it Contact geometric descriptions of vector fields on dually flat spaces and their applications in electric circuit models and nonequilibrium thermodynamics }, 
J. Math. Phys., {\bf 57}, 102702 [ 40 pages ], (2016).
 
\bibitem{Bravetti2015JPhysA}
A. Bravetti and C.S. Lopez-Monsalvo,
{\it Para-Sasakian geometry in thermodynamic fluctuation theory}  
J.Phys.A: Math. Theor. {\bf 48}, 125206 [ 21 pages ],  
(2015).

\bibitem{Bravetti2015}
A. Bravetti, C.S. Lopez-Monsalve and F. Nettel,
{\it Contact symmetries and Hamiltonian thermodynamics}, 
Ann. Phys., {\bf 361}, 377--400, (2015). 
  
\bibitem{GH2018}
S. Goto and H. Hino, 
{\it Information and contact geometric description of expectation variables exactly dericed from master equations},
arXiv:1805.10592v2.

\bibitem{Zamkovoy2009}
S. Zamkovoy, {\it Canonical connections on paracontact manifolds},
Ann. Glob. Geom. {\bf 36}, 37--60, (2009).
 
 
\bibitem{Silva2008}
A.C. da Silva, 
{\it Lectures on Symplectic Geometry}, 2nd Ed., 
Springer, (2008).

\bibitem{Mrugala1978}
R. Mrugala, 
{\it Geometrical formulation of equilibrium phenomenological thermodynamics}, 
Rep. Math. Phys., {\bf 14}, 419--427, (1978).

\bibitem{Mrugala1990}
R. Mrugala, {\it Statistical approach to the geometric structure of thermodynamics}, Phys. Rev. A, {\bf 41}, 3156--3160, (1990).  

\bibitem{Mrugala2000}
R. Mrugala, 
{\it On contact and metric structures on thermodynamic spaces}, 
Suken kokyuroku {\bf 1142}, 167--181, (2000). 







\end{thebibliography}
\end{document}